%% file: main.tex
\documentclass[copyright,creativecommons]{eptcs}
\providecommand{\event}{SSV 2012}

\usepackage{etex} 
\usepackage{xspace}
\usepackage{fixme}
\usepackage[backgroundcolor=orange]{todonotes}
\usepackage{amsmath, amsfonts, amssymb}
\usepackage{mathtools}
\usepackage[ruled,vlined,linesnumbered]{algorithm2e}
\usepackage[numbers,sectionbib,sort]{natbib}
\usepackage{cleveref}
\usepackage[extra]{tipa}
\usepackage{cases}
\usepackage{caption,subfig}
\usepackage{ctable}
\usepackage{textpos}
\usepackage{wrapfig}
\usepackage{pgf}
\usepackage{tikz}
\usepackage{pgfplots}
\pgfplotsset{every axis legend/.append style={
at={(1.02,1)},
anchor=north west}}
\pgfplotsset{every axis/.append style={line width=1pt}}

\usetikzlibrary{petri,arrows,backgrounds,fit,positioning,automata,matrix,decorations.pathreplacing,patterns}

\tikzset{
      biglabel/.style= {text width=2.5cm, font=\scriptsize\sffamily},
      biglabelleft/.style={biglabel,node distance=45},
      biglabelright/.style={biglabel,node distance=55},
      biglabelabove/.style={biglabel,node distance=55,text centered},
      biglabelbelow/.style={biglabel,node distance=55,text centered}
      }
\tikzset{
      initial/.style={accepting},
      every state/.style={thick,minimum width=2mm}
      }
\tikzset{
      every transition/.style={fill,minimum width=2mm,minimum height=5.5mm},
      every token/.style={fill=lightgray, text=black,font=\scriptsize}
      }
\tikzset{
      arc/.style={->,shorten >=0.5pt,thick,>=stealth},
      transarc/.style={arc,>=diamond},
      inhibarc/.style={arc,>=o},
      resetarc/.style={arc,->>},
      pre/.style={<-,shorten <=0.5pt,>=stealth,thick},
      post/.style={->,shorten >=0.5pt,>=stealth,thick}
      }

\tikzset{
      pathstate/.style={place,fill=black,minimum size=1mm},
      }

\tikzset{help lines/.style={step=1cm,gray,very thin},
  axes/.style={->,>=triangle 60,semithick},
  edges/.style={->,font=\footnotesize}
}

\tikzstyle{level 1}=[level distance=2cm, sibling distance=3cm]
\tikzstyle{level 2}=[level distance=2cm, sibling distance=3cm]

\newtheorem{theorem}{Theorem}[section]
\newtheorem{lemma}[theorem]{Lemma}
\newtheorem{proposition}[theorem]{Proposition}

\newenvironment{proof}[1][Proof]{\begin{trivlist}
\item[\hskip \labelsep {\bfseries #1}]}{\end{trivlist}}
\newenvironment{definition}[1][Definition]{\begin{trivlist}
\item[\hskip \labelsep {\bfseries #1}]}{\end{trivlist}}

\newcommand{\qed}{\nobreak \ifvmode \relax \else
      \ifdim\lastskip<1.5em \hskip-\lastskip
      \hskip1.5em plus0em minus0.5em \fi \nobreak
      \vrule height0.75em width0.5em depth0.25em\fi}

\input{macros}

\begin{document}
%
%
\title{A Forward Reachability Algorithm for \\ Bounded Timed-Arc Petri 
Nets\footnote{The paper was partially supported by VKR Center 
of Excellence MT-LAB.}}

\author{Alexandre David \quad Lasse Jacobsen \quad Morten Jacobsen
\quad Ji\v{r}\'{\i} Srba
\institute{Department of Computer Science, Aalborg University,\\
Selma Lagerl\"{o}fs Vej 300, DK-9220 Aalborg East, Denmark}
}

\def\titlerunning{A Forward Reachability Algorithm for Bounded Timed-Arc
Petri Nets}
\def\authorrunning{A. David, L. Jacobsen, M. Jacobsen \& J. Srba}

\maketitle              
\begin{abstract}
Timed-arc Petri nets (TAPN) are a well-known time extension of the
Petri net model and several translations to networks of timed
automata have been proposed for this model.
We present a direct, DBM-based
algorithm for forward reachability 
analysis of bounded TAPNs extended with 
transport arcs, inhibitor arcs and age invariants.
We also give a complete proof of its correctness,
including reduction techniques based on symmetries and extrapolation.
Finally, we augment the algorithm with a novel state-space
reduction technique introducing a monotonic ordering on markings and prove 
its soundness even in the presence of monotonicity-breaking features like 
age invariants and inhibitor arcs. 
We implement the algorithm within the model-checker
TAPAAL and the experimental results document an encouraging performance 
compared to verification approaches that translate TAPN models 
to UPPAAL timed automata. 
\end{abstract}

\section{Introduction}
\input{intro-new}


\section{Timed-Arc Petri Nets}\label{sec:symbolic}
\input{preliminaries}

\input{tapn}
\input{soundness}

\section{Extrapolation via DBMs}\label{sec:ext-correctness}
\input{correctness}

\section{Monotonicity of Bounded TAPNs}\label{sec:discrete-inclusion}
\input{discrete-inclusion}

\section{Implementation of the Reachability Algorithm}\label{sec:implementation}
\input{algorithms}

\section{Experiments}\label{sec:experiments}

\input{experiments-short}

\section{Conclusion}
\input{conclusion}

\bibliographystyle{eptcs}
\bibliography{references}


\end{document}

%% file: macros.tex
\renewcommand{\cref}[1]{\Cref{#1}}

\newcommand{\C}{\ensuremath{\mathcal{C}}\xspace}

\newcommand{\R}{\ensuremath{\mathbb{R}}\xspace}
\newcommand{\Rpos}{\ensuremath{\R_{\geq 0}}\xspace}
\newcommand{\N}{\ensuremath{\mathbb{N}}\xspace}
\newcommand{\Z}{\ensuremath{\mathbb{Z}}\xspace}
\newcommand{\Nzero}{\ensuremath{\N_{0}}\xspace}

\newcommand{\allIntervals}{\ensuremath{\mathcal{I}}\xspace}

\newcommand{\defTAPN}{\ensuremath{(P, T, \mathit{IA}, \mathit{OA}, c, \mathit{Type}, \iota)}\xspace}

\newcommand{\lab}{\ensuremath{Lab}\xspace}
\newcommand{\AP}{\ensuremath{\mathit{AP}}\xspace}

\newcommand{\spre}[1]{\ensuremath{{}^\circ #1}\xspace}
\newcommand{\spost}[1]{\ensuremath{#1^\circ}\xspace}
\newcommand{\preset}[1]{\ensuremath{{}^\bullet #1}\xspace}
\newcommand{\postset}[1]{\ensuremath{#1^\bullet}\xspace}

\newcommand{\age}[1]{\ensuremath{\mathit{v}_{#1}}\xspace}
\newcommand{\place}[1]{\ensuremath{\mathit{pl}_{#1}}\xspace}
\newcommand{\marking}[1]{\ensuremath{(\place{#1}, \age{#1})}\xspace}
\newcommand{\symMarking}[1]{\ensuremath{(\place{#1}, \Val{#1})}\xspace}

\newcommand{\move}[3]{\ensuremath{\mathit{move}(#1{}, #2{}, #3{})}\xspace}

\newcommand{\0}{\ensuremath{\mathbf{0}}\xspace}
\newcommand{\up}{\ensuremath{^\uparrow}\xspace}
\newcommand{\reset}[1]{\ensuremath{^{#1=0}}\xspace}

\newcommand{\IN}{\ensuremath{\mathit{IN}}\xspace}
\newcommand{\inv}[1]{\ensuremath{I(#1{})}\xspace}
\newcommand{\g}[3]{\ensuremath{g(#1{}, #2{}, #3{})}\xspace}

\newcommand{\goes}[1]{\ensuremath{\stackrel{#1}{\longrightarrow}}\xspace}
\newcommand{\symtrans}[1]{\ensuremath{\stackrel{#1}{\leadsto}}\xspace}

\newcommand{\INC}[1]{\ensuremath{\mathit{INC}_{#1}}\xspace}
\newcommand{\BOT}[1]{\ensuremath{\mathit{bot}(#1)}\xspace}
\newcommand{\maxcname}{\ensuremath{\mathit{mc}}\xspace}
\newcommand{\maxc}[1]{\ensuremath{\maxcname(#1)}\xspace}
\newcommand{\maxcnameI}{\ensuremath{\mathit{mci}}\xspace}
\newcommand{\maxcI}[1]{\ensuremath{\maxcnameI(#1)}\xspace}

\newcommand{\abstrname}[1]{\ensuremath{\alpha_{#1}}\xspace}
\newcommand{\abstrnamep}[1]{\ensuremath{\alpha'_{#1}}\xspace}
\newcommand{\abstr}[2]{\ensuremath{\abstrname{#1}\left(#2\right)}\xspace}
\newcommand{\abstrp}[2]{\ensuremath{\abstrnamep{#1}\left(#2\right)}\xspace}
\newcommand{\extplace}[2]{\ensuremath{\mathit{ext}_{#1}\left(#2\right)}\xspace}
\newcommand{\ext}[1]{\extplace{\place{}}{#1}}
\newcommand{\powerset}[1]{\ensuremath{2^{#1}}\xspace}

\def\slbracket{[\mkern-4.5mu(}
\def\srbracket{)\mkern-4.5mu]}

\newcommand{\lbname}{\ensuremath{\mathit{lb}}\xspace}
\newcommand{\ubname}{\ensuremath{\mathit{ub}}\xspace}
\newcommand{\lb}[2]{\ensuremath{\lbname_{#1}(#2)}\xspace}
\newcommand{\ub}[2]{\ensuremath{\ubname_{#1}(#2)}\xspace}

\newcommand{\lbcname}{\ensuremath{\mathit{lb}^{\eta}}\xspace}
\newcommand{\ubcname}{\ensuremath{\mathit{ub}^{\eta}}\xspace}
\newcommand{\lbc}[2]{\ensuremath{\lbcname_{#1}(#2)}\xspace}
\newcommand{\ubc}[2]{\ensuremath{\ubcname_{#1}(#2)}\xspace}

\newcommand{\lbopname}{\ensuremath{\mathit{lb}^{\lhd}}\xspace}
\newcommand{\ubopname}{\ensuremath{\mathit{ub}^{\lhd}}\xspace}
\newcommand{\lbop}[2]{\ensuremath{\lbopname_{#1}(#2)}\xspace}

\newcommand{\Val}[1]{\ensuremath{W_{#1}}\xspace}
\newcommand{\setOfVal}{\ensuremath{\mathcal{\Val{}}^{\C}}\xspace}

\newcommand{\Normal}{\ensuremath{\mathit{Normal}}\xspace}
\newcommand{\Inhib}{\ensuremath{\mathit{Inhib}}\xspace}
\newcommand{\Transport}[1]{\ensuremath{\mathit{Transport}_{#1}}\xspace}

\newcommand{\maxgc}{\ensuremath{\mathit{gc}}\xspace}
\newcommand{\maxclock}{\ensuremath{n}\xspace}
\newcommand{\cnt}[3]{\ensuremath{\mathit{count}^{#2}_{#1}(#3)}\xspace}
\newcommand{\setoftokens}{\ensuremath{X}\xspace}

\newcommand{\markingsubseteq}[1]{\ensuremath{\sqsubseteq_{#1}}\xspace}

\newcommand{\inc}[1]{\ensuremath{\mathit{inc}(#1)}\xspace}
\newcommand{\eq}[1]{\ensuremath{\mathit{eq}(#1)}\xspace}
\newcommand{\setOfIncPlaces}{\ensuremath{P_{\mathit{inc}}}\xspace}

\newcommand{\EF}{\ensuremath{\operatorname{{\sf EF}}}\xspace}
\newcommand{\AG}{\ensuremath{\operatorname{{\sf AG}}}\xspace}

\newcommand\TSPACE{\rule{0pt}{2.6ex}}
\newcommand\BSPACE{\rule[-1.2ex]{0pt}{0pt}}
\newcommand{\PADROW}{\TSPACE\BSPACE}

\newcommand{\untimed}[1]{\ensuremath{\mathit{untimed}(#1)}\xspace}
\newcommand{\abstrsubseteq}{\ensuremath{\subseteq}\xspace}
\newcommand{\arbabstr}{\arbabstrsecond{\equiv}}
\newcommand{\arbabstrsecond}[1]{\ensuremath{\abstrname{\mathit{id}},\abstrname{#1}}\xspace}

%% file: intro-new.tex
Time-dependent models and their formal analysis have attracted
a considerable research activity.
Notable formalisms include timed 
automata (TA)~\cite{DBLP:journals/tcs/AlurD94}, 
time Petri nets (TPN)~\cite{PMMerlin1974}
and timed-arc Petri nets 
(TAPN)~\cite{DBLP:conf/pstv/BolognesiLT90}. 
A comparison of the different modelling formalisms is provided 
in~\cite{S:FORMATS:08}.

We shall focus on the TAPN model where tokens are assigned a nonnegative real
number representing their age and input arcs of transitions contain
time intervals restricting the usable ages of tokens for transition
firing. 
The state-space of the model is in general infinite in two dimensions: 
the number of tokens in a marking can be unbounded, and the continuous 
time aspect induces infinitely 
many clock 
valuations. Indeed, the reachability problem for the model
is undecidable~\cite{TAPN_reachability}, while coverability remains
decidable~\cite{AN:BQQ}. Moreover, for modelling purposes
additional features like inhibitor/transport arcs and
age invariants are needed but they cause the undecidability 
also of the coverability problem~\cite{memics_paper}.

We restrict our focus to bounded TAPNs where the maximum number of tokens in
all reachable markings is fixed.
This model is equally expressive to networks of timed 
automata~\cite{S:ICATPN:05}
and efficient translations from
TAPN into UPPAAL timed automata~\cite{JJMS:SOFSEM:11} 
have been implemented and employed
in the  model-checker TAPAAL~\cite{DJJJMS:TACAS:12}. The
translation approach has though some drawbacks:
experimentation with state-space reduction techniques is difficult and
the engine does not return error traces when symmetry reduction is enabled.

We therefore design a novel reachability algorithm for extended TAPN that
incorporates an efficient extrapolation, symmetry reduction and 
monotonic inclusion techniques to optimize its performance, while at
the same time returning error traces with concrete time delays.
We give a complete proof of the algorithm correctness,
including all the optimization techniques. 
We provide an efficient ({\tt C++}),
open-source implementation of the algorithm and integrate the new engine
into the tool TAPAAL.
The experiments confirm a high efficiency of the new reachability algorithm and
we document this by two larger case-studies.
 
\paragraph{Related work.}
Verification techniques for TAPNs include
a backward coverability algorithm based on
existential zones~\cite{AN:BQQ} (notably terminating
also for unbounded nets)
and a forward reachability algorithm based
on region generators presented in~\cite{ADMN:NJC:07}. 
Both algorithms rely on the monotonic behavior
of the generated transition systems, however,  
inhibitor arcs and age invariants break this monotonicity~\cite{memics_paper} 
and hence the techniques are not applicable for extended TAPNs.
Backward algorithms are generally rather 
inefficient for on-the-fly state-space exploration and for the
employment of state-space reductions while 
the forward algorithm from~\cite{ADMN:NJC:07} is based
on a less efficient region construction instead of a zone-based one.
The algorithms were implemented in prototype
tools with no GUI and are not maintained any more.

There are efficient tools like TINA~\cite{tina} or Romeo~\cite{romeo} 
for model-checking  
Time Petri nets (TPN). The tools are based on abstractions using 
state-class graphs 
but even
though bounded TPN are essentially equally expressive as bounded TAPNs
(see~\cite{S:FORMATS:08} for an overview), 
the translations are exponential
and do not allow to perform a direct performance comparison 
because the modelling capabilities and the treatment 
of time in TPN and TAPN are very different.

The definition of our extrapolation (abstraction) operator 
is following~\cite{DBLP:journals/sttt/BehrmannBLP06} where a similar operator
was suggested for timed automata; our extension (apart from its adaptation
to the TAPN setting) is the handling of dynamic maximum constants depending 
on the current marking (see also~\cite{HKSW:FSTTCS:11} for a dynamic extrapolation
on timed automata). The main novelty is our 
definition of an inclusion operator that incorporates symmetry reduction
and works also for nets with monotonicity-breaking features.

%% file: preliminaries.tex
Let $\N$
be the set of natural numbers and
let $\Nzero = \N \cup \left\{0\right\}$. By 
$\Rpos$ we denote the set of non-negative real numbers.
The set of 
time intervals $\mathcal{I}$ 
is given by the abstract syntax 
($a \in \N_0, b \in \N$  and $a < b$):
$I ::= [a, a] \, \mid \, [a, b]  \, \mid \,  [a, b)  \, \mid \,   (a, b]  \, \mid \,   (a, b)  \, \mid \,   [a, \infty)  \, \mid \,  (a, \infty)$. 
The set of invariant intervals, $\mathcal{I}_{\text{Inv}}$, consists
of intervals that include $0$.

Let $\C = \{\0,1,2, \ldots, \maxclock\}$ be a finite set of real-valued 
clocks whose elements (numbers) represent names of clocks. 
The clock $\0$ is a special pseudoclock that has always the value 0.
A \emph{(clock) valuation} over $\C$ is a function 
$\age{} : \C \rightarrow \Rpos$
such that $\age{}(\0) = 0$. 
The set of all valuations over the clocks $\C$ is denoted by $\setOfVal$.
Let $\age{}$ be a valuation and $d$ a nonnegative real. 
We let $\age{} + d$ be the valuation such that 
$(\age{} + d)(i) = \age{}(i) + d$ for every $i \in \C \setminus \{\0\}$ 
and $(\age{}+d)(\0) = 0$. Further, for a subset of clocks $R \subseteq \C$, 
we let $\age{}\reset{R}$ be the valuation such that $\age{}\reset{R}(i) = 0$ 
if $i \in R$ and $\age{}\reset{R}(i) = \age{}(i)$ otherwise.

\input{figures/pairing-example}

Let $\Val{} \subseteq \setOfVal$ be a set of valuations and let
$R \subseteq \C$.
We define the \emph{delay operation} as 
$\Val{}\up = \{ \age{}+ d \mid \age{} \in \Val{} \text{ and } d \in \Rpos\}$
and the \emph{reset operation} as 
$\Val{}\reset{R} = \{ \age{}\reset{R} \mid \age{} \in \Val{} \}$.


A \emph{timed labeled transition system} (TLTS) is a 
tuple $(S, \lab, \goes{})$ where
$S$ is a set of states (or processes),
$\lab = \mathit{Act} \cup \Rpos$ is a set of labels, consisting of discrete actions and time delays, and
$\goes{}\, \subseteq (S \times \lab \times S)$ is the transition relation.
We often write $s \goes{\alpha} s'$ instead of $(s, \alpha, s') \in \goes{}$
and if the label is not important, we simply write $s \goes{} s'$.

%% file: figures/pairing-example.tex
\begin{wrapfigure}{r}{5.6cm}
\vspace{-3mm}
\begin{tikzpicture}[yscale=1,xscale=0.82]
\node[place,label=above:$p_1$,structured tokens={2.1}] at (-2,1) (p1) {};
\node[place,label=below:$p_2$,structured tokens={3.4}] at (-2,-1) (p2) {};
\node[place,label=0:$p_3$] at (2,1.5) (p3) {};
\node[place,label=-4:inv: $<3$,label=10:$p_4$] at (2,0.5) (p4) {};
\node[place,label=0:$p_5$] at (2,-0.5) (p5) {};
\node[place,label=0:$p_6$] at (2,-1.5) (p6) {};
\node[transition, label=above:$t$] at (0,0) (t) {};

\draw[transarc] (p1) -- (t) node[above,midway,sloped] {$[2,3]$};
\draw[arc] (p2) -- (t) node[below,midway,sloped] {$(1,6]$};
\draw[arc] (t) -- (p3) {};
\draw[transarc] (t) -- (p4) node[above,pos=0.75,sloped] {};
\draw[arc] (t) -- (p5) {};
\draw[arc] (t) -- (p6) {};
\end{tikzpicture}
\caption{A TAPN with $Pairing(t) = \{ (p_1,p_4), (p_2,p_3), (\bot, p_5), (\bot, p_6) \}$}\label{fig:pairing-example}
\vspace{-5mm}
\end{wrapfigure}
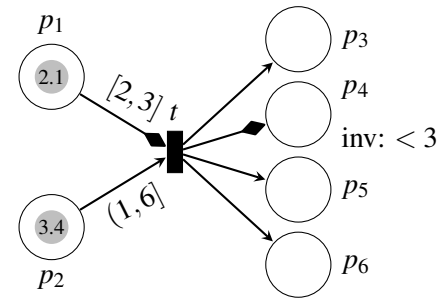

%% file: tapn.tex

We shall now define the Timed-Arc Petri Net (TAPN) model, restricting ourselves
to $k$-bounded nets (where every reachable marking has at most $k$ tokens).
An example of a 4-bounded 
TAPN is given in \cref{fig:pairing-example}.
It consists of six places (circles), one transition (rectangle)
and two tokens of age 2.1 and 3.4 representing the current marking.
Input arcs to the transition $t$ contain time intervals and because both
tokens belong to the corresponding interval, the transition can fire,
consume the two tokens in $p_1$ and $p_2$, 
and produce a new token of age $0$ to each of the 
places $p_3$, $p_5$ and $p_6$.
Because the place $p_1$ is connected to $p_4$ via a pair of
transport arcs (denoted by a diamond tip), 
the token of age 2.1 is moved to $p_4$ while its age
is preserved. Should there be more pairs of transport arcs connected to
the transition $t$, we label them with numbers so that the 
routes on which tokens travel are clearly marked. Finally, note that
the place $p_4$ has an associated age invariant, restricting the possible
ages of tokens in the place to strictly less than $3$. Should we in
the current marking first delay 0.9 time units, both tokens in $p_1$ and $p_2$
would 
still fit into their intervals but the transition $t$ is not enabled any more
due to the age invariant in the place $p_4$.



\begin{definition}\label{def:tapn}
A TAPN is a 7-tuple $\defTAPN$ where
\begin{itemize}
\item $P$ is a finite set of places,
\item $T$ is a finite set of transitions such that $P \cap T = \emptyset$,
\item $\mathit{IA} \subseteq P \times T$ is a finite set of input arcs,
\item $\mathit{OA} \subseteq T \times P$ is a finite set of output arcs,
\item $c : \mathit{IA} \rightarrow \mathcal{I}$ assigns intervals to input arcs, 
\item $\mathit{Type} : \mathit{IA} \cup \mathit{OA} \goes{} 
\{\Normal, \Inhib \} \cup \{ \Transport{i} \mid i\in\N \}$
is a function assigning a type to all arcs such that
	\begin{itemize}
	\item $\mathit{Type}(a) = \Inhib \Rightarrow a \in \mathit{IA} \wedge c(a) = [0,\infty)$,
	\item $\mathit{Type}(p,t) = \Transport{\ell} \Rightarrow \exists ! (t,p') \in \mathit{OA} \, . \, \mathit{Type}(t,p') = \Transport{\ell}$ and
	\item $\mathit{Type}(t,p') = \Transport{\ell} \Rightarrow \exists ! (p,t) \in \mathit{IA} \, . \, \mathit{Type}(p,t) = \Transport{\ell}$, and
	\end{itemize}
\item $\iota : P \rightarrow \mathcal{I}_{\mathit{inv}}$ assigns age invariants to places.
\end{itemize}
\end{definition}

For notational convenience, we write $\mathit{Type}(a) = \Transport{}$ 
if $\mathit{Type}(a) = \Transport{\ell}$ for some $\ell$.
For a transition $t \in T$, we define the \emph{preset} of $t$ as 
$\preset{t} = \{ p \in P \mid (p, t) \in \mathit{IA}, 
\mathit{Type}(p,t) \neq \Inhib\}$ and the \emph{postset} of $t$ as 
$\postset{t} = \{ p \in P \mid (t, p) \in \mathit{OA} \}$. 

We denote by $P_{\bot}$ the set $P \cup \{\bot\}$ where
$\bot$ is a special symbol representing a pseudo-place that holds
the currently unused tokens.
The \emph{augmented preset} and \emph{augmented postset} 
of a transition $t$ are defined as the multisets
\begin{align*}
\spre{t} &= \{ p_1, \ldots, p_m \mid \{p_1, \ldots, p_\ell \} = \preset{t}, p_i = \bot \text{ if } \ell < i \leq m \} \\
\spost{t} &= \{ p_1, \ldots, p_m \mid \{p_1, \ldots, p_\ell \} = \postset{t}, p_i = \bot \text{ if } \ell < i \leq m \} 
\end{align*}
where $m = \max(|\preset{t}|,|\postset{t}|)$.
This guarantees that $|\spre{t}| = |\spost{t}|$ for any transition $t$,
a convenient technical detail used in the algorithms.
We also extend the definition of $c$ and $\iota$ such that 
$c(\bot, t) = [0, \infty)$ whenever $\bot \in \spre{t}$ and 
$\iota(\bot) = [0,\infty)$. 

A \emph{token} in a $k$-bounded TAPN is an element from the set 
$\{1, 2, \ldots, k\}$. A \emph{marking} is a pair $M = \marking{}$ where 
$\place{} : \{1,2,\ldots,k\} \rightarrow P_{\bot}$ is the placement function 
and $\age{} : \{1,2,\ldots,k\} \rightarrow \Rpos$ is the age function. 
The placement determines the current location of each token
(it returns $\bot$ if the token is unused) and the age function represents
the age of each token.
The placement function will be sometimes written as a vector where 
e.g. $[p_1,p_2,p_1]$ represents the fact that tokens 1 and 3 are located in 
the place $p_1$ and token 2 is located in $p_2$. 
The set of all markings on a $k$-bounded TAPN $N$ is denoted by 
$\mathcal{M}(N)$. A \emph{marked} $k$-bounded TAPN is a pair 
$(N, \marking{0})$ where $N$ is a $k$-bounded TAPN and 
$\marking{0}$ is the initial marking where $\age{0}(i)=0$ for all $i$, 
$1 \leq i \leq k$.

Since there are always $k$ tokens in any marking (unused ones are
in $\bot$), it is for algorithmic purposes convenient to fix for each 
transition the paths from input to output places.
This is formalized in the function
$\mathit{Pairing} : T \rightarrow 2^{P_\bot \times P_\bot}$ 
such that for every transition $t$ we have
\begin{align*}
Pairing(t) =\{ &(p_1, p'_1), \ldots, (p_\ell, p'_\ell) \mid \{p_1, \ldots, p_\ell \} = \spre{t}, \{p'_1, \ldots, p'_\ell \} = \spost{t} \text{ and } \\
&\mathit{Type}(p_i,t) = \mathit{Type}(t,p'_j) = \Transport{\ell} \Rightarrow i = j \} \, .
\end{align*}
An example of a possible pairing function is given 
in \cref{fig:pairing-example}.

The effect of moving tokens in a placement $\place{}$ 
by firing a transition $t$ with the pairing 
$\mathit{Pairing}(t) = \{(p_1,p_1'),$ $(p_2,p_2'), \ldots, (p_\ell, p_\ell')\}$
is defined in the expected way as follows.
Let $\IN = \{i_1,i_2,\ldots, i_\ell\} \subseteq \{1,2,\ldots,k\}$ 
be a set of tokens placed in the places $p_1$ to $p_\ell$ and used
for firing $t$.
Formally, $\place{}(i_j) = p_j$ for all $1 \leq j \leq \ell$. 
The move function $\move{\place}{\IN}{t} : \{1,2,\ldots,k\} \rightarrow P_\bot$ 
is now given by
\begin{align*}
 \move{\place}{\IN}{t}(i) = \begin{cases}
	\place{}(i)  & \text{if } i \notin \IN  \\
	p_j' & \text{if } i \in \IN \text{ such that } i = i_j \, . \\ 
\end{cases}
\end{align*}


Consider \cref{fig:pairing-example} and let $\place{} = \begin{bmatrix} p_1, p_2, \bot, \bot \end{bmatrix}$. Then 
$\move{\place{}}{\{1,2,3,4\}}{t} = \begin{bmatrix} p_4, p_3, p_5, p_6 \end{bmatrix}$.

\begin{definition}[Transition Enabledness]\label{def:tapn_concrete_enabledness}
A transition $t \in T$ is \emph{enabled} by a set of tokens $\IN \subseteq \{1,2, \ldots, k\}$ in a marking $\marking{}$ if 
\begin{enumerate}
\item[(i)] $\spre{t} = \{ \place{}(i) \mid i \in \IN\}$
\item[(ii)] $\age{}(i) \in c(\place{}(i),t)$ for all $i \in \IN$
\item[(iii)] $\mathit{Type}(\place{}(i),t) = \Transport{}$ implies 
$\age{}(i) \in \iota(move(\place{}, \IN, t)(i))$ for all $i \in \IN$
\item[(iv)] $(\place{}(i),t) \in \mathit{IA}$ implies 
  $\mathit{Type}(\place{}(i),t) \neq \Inhib$ for all
$i \in \{1,2,\ldots,k\} \setminus \IN$.
\end{enumerate}
\end{definition}
A transition $t$ is hence enabled if there is a token in each of its
input places \emph{(i)}, the ages of these tokens fit into the
intervals on the input arcs \emph{(ii)}, the age of the token that is moved
along a pair of transport arcs does not break the age invariant of the
place where is it moved to \emph{(iii)}, and there is no token in 
any place connected via inhibitor arc to the transition $t$ \emph{(iv)}.
\begin{definition}[Transition Firing]
A transition $t$ enabled in a marking $\marking{}$ by the  set of
tokens $\IN$ can \emph{fire}, 
producing a marking $(\move{\place}{\IN}{t}, \age{}\reset{R})$ where 
$R = \{ i \in \IN \mid \mathit{Type}(\place{}(i),t) \neq \Transport{} \}$.
\end{definition}
\begin{definition}[Time Delay]
A time delay of $d \in \Rpos$ time units is possible
in a marking $\marking{}$ if $\age{}(i) + d \in \iota(\place{}(i))$ 
for all $i \in \{1,2,\ldots, k\}$. By delaying $d$ time units, 
we reach the marking $(\place{}, \age{} + d)$.
\end{definition}

The \emph{concrete execution semantics} of a 
TAPN $N = \defTAPN$ is given by a 
TLTS $T(N) = (\mathcal{M}(N), T \cup \Rpos, \longrightarrow)$ 
where states are markings on $N$ and labels 
are transition names and time delays. 
The transition relation $\goes{}$ is defined so that $M \goes{t} M'$ 
if by firing $t$ in the marking $M$ we reach the marking $M'$,
and $M \goes{d} M'$ if by delaying $d$ time units in the marking $M$ 
we reach the marking $M'$. 

\section{Symbolic Semantics}
The concrete execution semantics is not suitable for the actual verification as
there are infinitely (in fact uncountably) many reachable markings. 
Therefore we give a symbolic semantics of $k$-bounded TAPNs with respect
to some given abstraction operator and show that the symbolic 
semantics preserves the answer to the reachability question.

A \emph{symbolic marking} of a $k$-bounded TAPN is a pair $(\place{}, \Val{})$ where $\place{} : \{1,2,\ldots,k\} \rightarrow P_{\bot}$ is a placement function and $\Val{} \subseteq \setOfVal$ is a set of valuations. 

In order to guarantee the finiteness of the state-space in the
abstract semantics, we consider abstraction operators that can
enlarge (extrapolate) the possible sets of valuations in symbolic markings. 
Instead of considering global abstraction operators like for example
in the timed automata theory 
(see e.g.~\cite{DBLP:journals/sttt/BehrmannBLP06}), our abstraction
operators depend also on the current placement.

\begin{definition}
An \emph{abstraction operator} is a function $\abstrname{} : [\{1,2,\ldots,k\} \rightarrow P_{\bot}] \times \powerset{\setOfVal} \goes{} \powerset{\setOfVal}$ 
such that $\Val{} \subseteq \abstr{}{\place{}, \Val{}}$
for all symbolic markings $(\place{}, \Val{})$.
\end{definition}

An example of an abstraction operator is 
the identity abstraction operator $\abstrname{\mathit{id}}$ 
where $\abstr{id}{\place{}, \Val{}} = \Val{}$ for all symbolic 
markings $(\place{},\Val{})$. 

Our aim is of course to find an operator that for a given
net abstracts as much
as possible. To do so, we use the function 
$\maxcnameI : \allIntervals \rightarrow \N_0$ that returns, for an interval $I$, the maximum constant different from $\infty$ appearing in $I$. 
Let $\maxgc$ be the maximum constant different from $\infty$ that appears 
in intervals or invariants of the given TAPN.
The function $\maxcname : P_{\bot} \rightarrow \N_0$ now returns, 
for each place $p$, the maximum constant appearing in the guards of 
outgoing arcs from $p$ or in the invariant of $p$; if there are transport
arcs connected to $p$, the constant is $\maxgc$. 
\begin{align*}
\maxc{p} = \begin{cases}
 \maxgc 
 \text{\ \ \ \ \ \ \ \ \ \ \ \ \ \ \ \ \ if 
 there exists $(p,t) \in \mathit{IA}$ s.t. 
  $\mathit{Type}(p, t) = \Transport{}$} \\
 \displaystyle\max\left( \maxcI{\iota(p)}, \max_{(p,t) \in \mathit{IA}}\left( \maxcI{c(p,t)} \right) \right)\text{ \ \ \ \ otherwise} .
 \end{cases}
\end{align*}\label{eqn:maxc}

Following \cite{DBLP:journals/sttt/BehrmannBLP06}, we proceed to define an equivalence on valuations. The addition in our paper
is that we take the placement function into account, 
thereby allowing for dynamic maximum constants.
Let $\place{}$ be a placement function and let $\age{}$ and $\age{}'$ be valuations. We write $\age{} \equiv_{\place{}} \age{}'$ if for all $i \in \C \setminus \{ \0 \}$
\begin{enumerate}
\item $\age{}(i) = \age{}'(i)$, or  
\item $\age{}(i) > \maxc{\place{}(i)}$ and $\age{}'(i) > \maxc{\place{}(i)}$.
\end{enumerate}

Hence two related valuations are indistinguishable from each other in the sense that they can be used to fire the same transitions. 
Now we can define an abstraction operator based on the relation above.
\begin{definition}
Let $\abstr{\equiv}{\place{},\Val{}} = \{ \age{}' \mid \age{}' \equiv_{\place{}} \age{} \text{ and } \age{} \in \Val{}\}$ for a set of valuations $\Val{} \subseteq \setOfVal$ and a placement function $\place{}$.
\end{definition}

Clearly, $\Val{} \subseteq \abstr{\equiv}{\place{},\Val{}}$ for any set of valuations $\Val{} \subseteq \setOfVal$ and any placement function $\place{}$
as the relation is reflexive.
For two abstraction operators $\abstrname{}$ and $\abstrnamep{}$
we write $\abstrname{} \abstrsubseteq \abstrnamep{}$ if 
$\abstr{}{\place{}, \Val{}} \subseteq \abstrp{}{\place{}, \Val{}}$ 
for all placement functions $\place{}$ and all $\Val{} \subseteq \setOfVal$.

We are now ready to give the symbolic semantics of TAPNs.
Let $g$ be a function that takes a placement function $\place{}$,
a set of tokens $\IN$ and a transition $t$ as its arguments
(assuming that $\spre{t} = \{\place{}(i) \mid i \in \IN\}$) and 
it returns the set of all valuations such that the tokens in $\IN$ 
satisfy all guards on the input arcs of $t$. 
Formally, $\g{\place}{\IN}{t} = \bigcap_{i \in \IN} \{\age{} \in \setOfVal 
\mid \age{}(i) \in c(\place{}(i), t)\}$. 
Similarly, we define a function $I$ that takes a placement function 
as its argument and returns the set of all valuations satisfying the
age invariants. Formally, 
$\inv{\place{}} = \bigcap_{i \in \C \setminus \{ \0 \}} 
\{\age{} \in \setOfVal \mid \age{}(i) \in \iota(\place{}(i))\}$.

\begin{definition}[Symbolic Semantics]\label{def:tapn_sym_semantics}
Let $(N, \marking{0})$ be a marked $k$-bounded TAPN and let $\abstrname{}$ be an abstraction operator. The symbolic semantics of $(N, \marking{0})$ is given by a TLTS $T(N) = (S, L, \symtrans{}_{\abstrname{}})$ where
\begin{itemize}
\item $S = [\{1,2,\ldots,k\} \rightarrow P_{\bot}] \times (\powerset{\setOfVal} \setminus \emptyset)$,
\item $L = T \cup \{\epsilon\}$, and
\item $\symMarking{} \symtrans{t}_{\abstrname{}} (\place{}', \abstr{}{\place{}', \Val{}'})$ if $t$ is a transition and  there is a set of tokens $\IN$ such that
\begin{itemize}
	\item $\spre{t} = \{\place{}(i) \mid i \in \IN\}$
	\item $\place{}' = \move{\place}{\IN}{t}$
	\item $\Val{}' \stackrel{\mathit{def}}{=} (\Val{} \cap \g{\place}{\IN}{t})\reset{R} \cap \inv{\place{}'}$ is consistent ($W'\not=\emptyset$)
where $R = \{ i \in \IN \mid \mathit{Type}(\place{}(i), t) \neq \Transport{} \}$
	\item 
$(\place{}(i),t) \in \mathit{IA}$ implies 
$\mathit{Type}(\place{}(i),t) \neq \Inhib$
for all $i \in \{ 1,\ldots,k \} \setminus \IN$
\end{itemize}
\item $\symMarking{} \symtrans{\epsilon}_{\abstrname{}} (\place{}, \abstr{}{\place{}, \Val{}\up \cap \inv{\place}})$.
\end{itemize}
The initial symbolic marking is $(\place{0}, \{\age{0}\})$ where $\age{0}(i) = 0$ for all $i \in \C$.
\end{definition}


%% file: soundness.tex
Let us define $\symtrans{T}_{\abstrname{}} \stackrel{\mathit{def}}{=} 
\cup_{t \in T} \symtrans{t}_{\abstrname{}}$.
We can now state the main theorem of this section, which establishes soundness and completeness of the symbolic semantics for any abstraction operator
between $\abstrname{id}$ and $\abstrname{\equiv}$.
In fact, we allow to dynamically change the abstraction operators
during a computation in the symbolic semantics.
Hence we consider a new transition relation
$\symtrans{}_{\arbabstrsecond{\mathit{\equiv}}} \stackrel{\mathit{def}}{=} \bigcup_{\abstrname{\mathit{id}} \abstrsubseteq \abstrname{} \abstrsubseteq \abstrname{\mathit{\equiv}}} \symtrans{}_{\abstrname{}}$ allowing us to apply in any step
an arbitrary 
abstraction operator between the identity and $\abstrname{\equiv}$.

\begin{theorem}\label{thm:soundness_completeness}
Let $(N, (\place{0}, \age{0}))$ be a marked $k$-bounded TAPN. Then
\begin{itemize}
\item \emph{(Soundness)} $(\place{0}, \{\age{0}\}) \symtrans{}^{*}_{\arbabstr} (\place{}, \Val{})$ implies that there exists a valuation $\age{} \in \Val{}$ 
such that $\marking{0} \goes{}^{*} (\place{},\age{})$, and
\item \emph{(Completeness)} $\marking{0} \goes{}^* \marking{}$ implies,
for any abstraction operator $\abstrname{}$ where 
$\abstrname{id} \abstrsubseteq \abstrname{} \abstrsubseteq \abstrname{\equiv}$, 
that $(\place{0}, \{ \age{0} \}) \symtrans{\epsilon}_{\abstrname{}} \circ \, \, (\symtrans{T}_{\abstrname{\mathit{id}}} \circ \symtrans{\epsilon}_{\abstrname{}})^* \,\, (\place{}, \Val{})$ for some $\Val{}$ where $\age{} \in \Val{}$.
\end{itemize}
\end{theorem}

Note that the completeness part of the theorem imposes that
the symbolic semantics can reach the given placement via a strictly alternating
sequence of time elapsing and transition firing steps where the transition 
firing steps are not extrapolated (using the identity abstraction 
operator); this reflects how the successors are computed in the 
reachability algorithm discussed in Section~\ref{sec:implementation}.

%% file: correctness.tex

For the use in our reachability algorithm, we need to
represent infinite sets of valuations $W$ in a finite way. However, it is
not known how to effectively deal directly with the $\abstrname{\equiv}$
abstraction operator. Instead, we suggest a slightly less general
abstraction (extrapolation) operator
and a way to finitely represent infinite sets of valuations
in order to guarantee a finite and effectively searchable
state-space of symbolic markings.

For this purpose we use Difference Bound Matrices (DBM),
a well-known technique for verification of real-time systems
(see e.g. \cite{DBLP:conf/avmfss/Dill89, DBLP:conf/ac/BengtssonY03})
that allows us to store constraints on
single clocks and on differences of two clocks in a compact matrix-based
data structure.

\begin{definition}[Difference Bound Matrix (DBM)]\label{def:dbm}
A \emph{Difference Bound Matrix} $D$ over the set of clocks $\C$ is a $|\C| \times |\C|$ matrix such that 
$$D_{ij} \in (\Z \times \{<, \leq\}) \cup \{(\infty, <)\}$$ 
where $i,j \in \C$ and for all $i \in \C$ we have 
\begin{enumerate}
\item if $D_{\0i} = (m, \lhd)$ then  $m \leq 0$, and $\lhd \in \{ <, \leq \}$, \label{itm:dbm_def_lb}
\item if $D_{i\0} = (m, \lhd)$ then $m \geq 0$, and $\lhd \in \{ <, \leq \}$, and \label{itm:dbm_def_ub}
\item $D_{ii} = (0, \leq)$. \label{itm:dbm_def_dc}
\end{enumerate}
\end{definition}
A \emph{solution} to a DBM $D$ is a valuation $v$ such that for all 
$i,j \in \C$ we have $v(i) - v(j) \lhd m$ where $D_{ij} = (m, \lhd)$.
The set of all solutions to a DBM $D$ (alternatively, the zone over $D$) 
is denoted by $[D]$.

We refer to the elements $D_{ij}$
as \emph{bounds}. 
A bound $D_{\0i} = (m,\lhd)$ where $m \leq 0$ 
(by Condition~\ref{itm:dbm_def_lb}) 
is called the \emph{lower bound} for the clock $i$.
Such a constraint means $v(\0) - v(i) \lhd m$ for any valuation $v \in [D]$, which is equivalent to 
$-m \lhd v(i)$.
Similarly, a bound $D_{i\0} = (m, \lhd)$ where $m \geq 0$ (by 
Condition~\ref{itm:dbm_def_ub}) 
is called the \emph{upper bound} for the clock $i$ and it 
means that $v(i) - v(\0) \lhd m$ which is the same as $v(i) \lhd m$. Finally, a bound $D_{ij}$ where $i \neq \0 \neq j$ is called a \emph{diagonal constraint}.

For notational convenience, we introduce an alternative notation 
$\lbname$ and $\ubname$ for the lower and upper bound of a clock $i$.
Formally, $\lb{D}{i} = (-m,\lhd)$ if $D_{\0i} = (m,\lhd)$ 
and $\ub{D}{i} = D_{i\0}$. We further define a notation for the 
individual elements in a bound such that $\lbc{D}{i} = m$ 
and $\lbop{D}{i} = \lhd$ if $\lb{D}{i} = (m,\lhd)$. We use the same notation 
$\ubcname_D$ and $\ubopname_D$ also for upper bounds.




A DBM $D$ is \emph{consistent} if $[D] \neq \emptyset$. 
We say that $D$ is in \emph{canonical form} 
if $D_{ij} \preceq D_{ik} + D_{kj}$ for all $i,j,k \in \C$. 
It is well known that
for every consistent DBM $D$ there is a unique canonical DBM $D^c$ 
such that $[D] = [D^c]$~\cite{DBLP:conf/avmfss/Dill89}.


We now define a variant of one of the abstraction (extrapolation) 
operators on DBMs 
in order to abstract sets of valuations represented by a DBM. The
definition is inspired by~\cite{DBLP:journals/sttt/BehrmannBLP06}, 
the main difference being the use of dynamic maximum constants in our 
operator.

\begin{definition}[Extrapolation]\label{def:ext_operator}
The extrapolation of a canonical DBM $D$ in a placement $\place{}$ is the DBM $D'$, called $\ext{D}$, and defined  as follows (here $i,j \in \C \setminus \{\0\}$ such that $i \neq j$):
\begin{enumerate}
\item $D' := D$
\item \label{itm:ext_def_lb_greater_than_c} if $\maxc{\place{}(i)} < \lbc{D}{i}$ then
$\lb{D'}{i} := (\maxc{\place{}(i)},<)$ and
$\ub{D'}{i} := (\infty,<)$
\item \label{itm:ext_def_ub_greater_than_c} if $\ubc{D}{i} > \maxc{\place{}(i)}$ then $\ub{D'}{i} := (\infty,<)$
\item \label{itm:ext-case4} if $\maxc{\place{}(i)} < \lbc{D}{i}$ or $\maxc{\place{}(j)} < \lbc{D}{j}$ then $D'_{ij} := (\infty,<)$
\item \label{itm:ext-case5} if $D_{ij} = (m,\lhd)$ and $m > \maxc{\place{}(i)}$ then $D'_{ij} := (\infty,<)$
\end{enumerate}
\end{definition}

Intuitively, the extrapolation works by removing all upper bounds greater 
than the maximum constant of a given place and by replacing any lower bound 
greater than the maximum constant with the value $(\maxc{\place{}(i)},<)$. 
Additionally, whenever the lower bound is above the maximum constant 
of a given place, any diagonal constraint involving that clock are also 
removed. An example of a DBM $D$ and its extrapolation $\ext{D}$
together with their graphical representations (clock~1 is on
the x-axis and clock~2 on the y-axis) is given in 
Figure~\ref{fig:extrapolation_example}. We can see that the extrapolation
operator enlarges the set of valuations represented by $D$ such that 
there are only finitely many extrapolated DBMs.
\input{figures/extrapolation_example}



\begin{lemma}\label{lem:finite_number_of_extrapolated_dbms}
The set 
$\{\ext{D} \mid D \text{ is a canonical DBM}\}$
is finite.
\end{lemma}

We can now conclude with the main result stating
that the extrapolation provides an abstraction which is between
identity and $\abstrname{\equiv}$; a crucial and nontrivial fact needed for
proving correctness of the reachability algorithm.

\bigskip
\begin{theorem}
\label{lem:extrapolation_below_abstraction_equiv}
Let $D$ be a canonical DBM and let $\place{}$ be a placement function. 
Then $[D] \subseteq [\ext{D}] \subseteq \abstr{\equiv}{\place{}, [D]}$.
\end{theorem}

%% file: figures/extrapolation_example.tex
\begin{figure}[tp]
\centering
\subfloat[][Canonical DBM $D$]{
\scalebox{0.9}{
\begin{tabular}[b]{c|c|c}
$(0, \leq)$ \PADROW & \PADROW $(-1, \leq)$ \PADROW & \PADROW $(-3, \leq)$ \\
\hline
$(5, \leq)$ \PADROW & \PADROW $(0, \leq)$ \PADROW & \PADROW $(1, \leq)$ \\
\hline
$(6, \leq)$ \PADROW & \PADROW $(3, \leq)$ \PADROW & \PADROW $(0, \leq)$
\end{tabular}
\label{fig:extrapolation_example_original}
}
}\quad\quad
\subfloat[][The DBM $\ext{D}$]{
\scalebox{0.9}{
\begin{tabular}[b]{c|c|c}
$(0, \leq)$ \PADROW & \PADROW $(-1, \leq)$ \PADROW & \PADROW $(-2, <)$ \\
\hline
$(\infty, <)$ \PADROW & \PADROW $(0, \leq)$ \PADROW & \PADROW $(\infty, <)$ \\
\hline
$(\infty, <)$ \PADROW & \PADROW $(\infty, <)$ \PADROW & \PADROW $(0, \leq)$
\end{tabular}
\label{fig:extrapolation_example_ext_D}
}}\\
\subfloat[][The zone ${[D]}$]{
\begin{tikzpicture}[scale=0.35]
    \draw[help lines] (0,0) grid (7,7);
    \draw[axes] (0,0) -- (8,0);
	\draw[axes] (0,0) -- (0,8); 
 	\filldraw[fill=gray!50, draw=black,thick] (1,3) -- (1,4) -- (3,6) -- 
	(5,6) -- (5,4) -- (4,3) -- (1,3);
	\node at (4, -1) {};
	\node at (-1, 4) {};
\end{tikzpicture}
\label{fig:extrapolation_example_original_zone}
}\quad\quad
\subfloat[][The zone ${[\ext{D}]}$]
{
\begin{tikzpicture}[scale=0.35]
      \draw[help lines] (0,0) grid (7,7);
    \draw[axes] (0,0) -- (8,0);
	\draw[axes] (0,0) -- (0,8); 
	\filldraw[fill=gray!50,draw=gray!50] (1,2) -- (1,7) -- (7,7) -- (7,2) -- (1,2); 
	\draw[draw=black, thick] (1,7) -- (1,2); 
	\draw[draw=black, dotted, thick] (1,3) -- (1,4) -- (3,6) -- (5,6) -- (5,4) -- (4,3) -- (1,3); 
	\draw[draw=black,dashed, thick] (1,-0.3) -- (1,7.3);
	\draw[draw=black,dashed, thick] (-0.3,2) -- (7.3,2);
	\node at (1, -1) {$\maxc{\place{}(1)}$};
	\node at (-2.5, 2) {$\maxc{\place{}(2)}$};
\end{tikzpicture}
\label{fig:extrapolation_example_ext_D_zone}
}
\caption{Example of the extrapolation operator for $\maxc{\place{}(1)} = 1, \maxc{\place{}(2)} = 2$}\label{fig:extrapolation_example}
\end{figure}
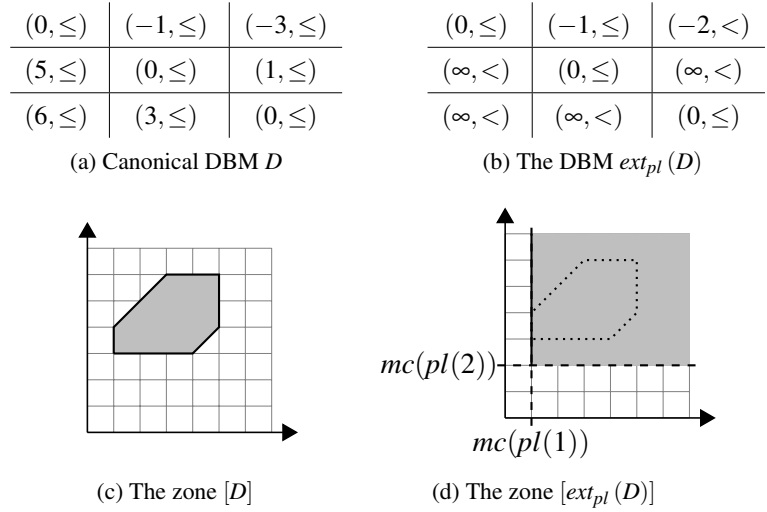

%% file: discrete-inclusion.tex
It is a well-known fact that the behaviour of the basic TAPN model
is monotonic~\cite{DBLP:journals/tcs/FinkelS01}
with respect to the standard marking inclusion,
intuitively meaning that adding more tokens to the net does not restrict 
its behaviour.
However, the use of age invariants and inhibitor arcs breaks the monotonicity
property~\cite{memics_paper}. 
In this section, we introduce a more refined inclusion relation 
on symbolic markings that preserves monotonicity
even in the presence of age invariants and inhibitor arcs.
Moreover, the inclusion relation allows for reordering of tokens
in the net and hence it implements the symmetry reduction.
The inclusion relation is then exploited in the reachability algorithm
presented in Section~\ref{sec:implementation}.


Let us fix a marked $k$-bounded TAPN $(N, (\place{0}, \age{0}))$. 
For a place $p \in P$, we define a boolean predicate 
$\untimed{p} \equiv 
(\iota(p) = [0,\infty)) \wedge \forall t \in \postset{p} \, . \, (\mathit{Type}(p,t) \neq \Transport{} \wedge c(p,t) = [0, \infty))$. If the predicate is true,
we do not need to keep track of the ages of tokens in this place.
For a symbolic marking $M$ we now define the set $\INC{M}$ 
representing the set of tokens eligible for the inclusion checking. 

\begin{definition}\label{def:inc}
Let $M = (\place{}, \Val{})$ be a symbolic marking. We define $\INC{M} \subseteq \{1,2,\ldots,k\}$ as the largest subset of tokens such that for any token $i \in \INC{M}$, 
\begin{enumerate}
\item \label{def:inc-bottom-condition}    $\place{}(i) \neq \bot$,
\item \label{def:inc-invariant-condition} $\iota(\place{}(i)) = [0, \infty)$,
\item \label{def:inc-inhib-condition}     $\place{}(i)$ has no outgoing inhibitor arcs, and
\item \label{def:inc-bound-condition}     either $\untimed{\place{}(i)}$ or 
	\begin{itemize}
	\item $\inf(V_i) \in V_i \Rightarrow \maxc{\place{}(i)} < \inf(V_i)$, or
	\item $\inf(V_i) \notin V_i \Rightarrow \maxc{\place{}(i)} \leq \inf(V_i)$,
	\end{itemize}
	where $V_i = \{ \age{}(i) \mid \age{} \in \Val{}\}$.
\end{enumerate}
\end{definition}
Let us briefly comment on Condition~\ref{def:inc-bound-condition}. 
If a place is untimed then the ages of tokens in that place are irrelevant 
and we can consider them for inclusion checking.
Otherwise, the lower bound of clock $i$ in $\Val{}$ is calculated 
by $\inf(V_i)$ 
and the two subconditions distinguish whether this bound is included or 
not\footnote{In the DBM representation we can read these
bounds directly from the matrix.}.
The point is that if the lower-bound for the token $i$ is above
the maximum constant for the place where $i$ is placed, then its concrete
age is irrelevant for the firing of transitions.



Let $\setOfIncPlaces \subseteq P$ be a set of places that
we want to consider for the inclusion checking (typically we set
$\setOfIncPlaces = P$ but the user can restrict
some places from the application of the inclusion operator by excluding
them from $\setOfIncPlaces$). 
We can now partition all tokens in the marking $M = (\place{}, \Val{})$ 
into three categories
\begin{itemize}
\item $\inc{M} = \INC{M} \cap \{i \mid \place{}(i) \in \setOfIncPlaces\}$
\item $\BOT{M} = \{i \mid \place{}(i) = \bot\}$
\item $\eq{M} = \{1,2,\ldots,k\} \setminus (\BOT{M} \cup \inc{M})$
\end{itemize}
where $\inc{M}$ contains all tokens eligible for inclusion checking,
$\BOT{M}$ contains all unused tokens and $\eq{M}$ is the set of all
tokens that have to be checked for equality.
Let us now introduce some notation. 
Let $\place{}$ be a placement function, $p$ a place and 
let $\setoftokens \subseteq \{1,2,\ldots,k\}$ be a set of tokens. 
We define $\cnt{\place{}}{\setoftokens}{p}$ = $|\{i \in \setoftokens \mid \place{}(i) = p\}|$. Intuitively, $\cnt{\place{}}{\setoftokens}{p}$ tells us how many tokens from $X$ are in the place $p$.
We are now ready to introduce the refined ordering relation.

\begin{definition}[Inclusion Ordering]\label{def:inclusion-op}
Let $M = (\place{}, \Val{})$ and $M' = (\place{}', \Val{}')$ be  
symbolic markings. We say that $M$ is included in $M'$, written $M \markingsubseteq{} M'$, if 
\begin{enumerate}
\item \label{def:inclusion-op-bijection-condition} There exists a bijection $h : \eq{M} \rightarrow \eq{M'}$ such that 
	\begin{enumerate}
	\item \label{def:inclusion-op-eq-places}  $\place{}(i) = \place{}'(h(i))$ for all $i \in \eq{M}$,
	\item \label{def:inclusion-op-valuations} for all $\age{} \in \Val{}$ there exists a $\age{}' \in \Val{}'$ such that for all $i \in \eq{M}$
	\begin{enumerate}
	\item[(i)] $\age{}(i) = \age{}'(h(i))$, or
	\item[(ii)] $\age{}(i) > \maxc{\place{}(i)}$ and $\age{}'(h(i)) > \maxc{\place{}'(h(i))}$,
	 \end{enumerate}
	\end{enumerate}
\item \label{def:inclusion-op-inc-count} $\cnt{\place{}}{\inc{M}}{p} \leq \cnt{\place{}'}{\inc{M'}}{p}$ for all $p \in P$.
\end{enumerate}
\end{definition}

Hence two symbolic markings $M$ and $M'$ are related by $\markingsubseteq{}$, 
if they agree on the sets $\eq{M}$ and $\eq{M'}$ via the bijection $h$ 
(this gives us the possibility to employ symmetry reduction), and moreover, 
the number of tokens in any 
place $p$ from the set $\inc{M'}$ in the marking $M'$ 
must be larger than 
or equal to the number of tokens in the place $p$ in the marking $M$.
We finish this section by a theorem proving monotonicity with respect
to the ordering relation $\markingsubseteq{}$ for any 
abstraction operator below $\abstrname{\equiv}$.

\begin{theorem}\label{thm:monotonicity_theorem}
Let $\abstrname{}$ be an abstraction operator such that 
$\abstrname{id} \subseteq \abstrname{} \subseteq \abstrname{\equiv}$ 
and $M_1, M_2 \in \mathcal{M}_{\abstrname{}}(N, (\place{0}, \{\age{0}\}))$ 
be reachable symbolic markings such that $M_1 \markingsubseteq{} M_2$. 
If $M_1 \symtrans{}_{\abstrname{}} M_1'$ then 
$M_2 \symtrans{}_{\abstrname{}} M_2'$ for some $M_2'$ such that
$M_1' \markingsubseteq{} M_2'$.
\end{theorem}

%% file: algorithms.tex
Before we present the reachability algorithm, let us first introduce a reachability fragment of 
CTL that is used in the algorithm. A formula of the logic is given by the abstract syntax: 
\begin{align}
\phi &::= \EF \psi \mid \AG \psi \nonumber\\
\psi,\psi_1,\psi_2 &::= (p \bowtie n) \mid 
\psi_1 \wedge \psi_2 \mid \psi_1 \vee \psi_2 \label{eqn:psi}
\end{align}
where $p \in P$, $n \in \Nzero$ and
$\bowtie \; \in \{<,\leq, =, \not=, \geq, >\}$.

The semantics of formulae is given in terms of a TLTS $(S, \lab, \goes{})$ and a labeling function $\mu : S \rightarrow 2^\AP$ assigning sets of true atomic propositions to states. 
We define the set of atomic propositions $\AP$ and the labeling function $\mu$ as $\AP \stackrel{\mathit{def}}{=} \{ (p \bowtie n) \mid p \in P, n \in \N_0 \text{ and } 
\bowtie \; \in \{ <, \leq, =, \not=, \geq, > \}\}$ and $\mu(M) \stackrel{\mathit{def}}{=} \{ (p \bowtie n) \mid \cnt{\place{}}{\{1,2,\ldots,k\}}{p} \bowtie n \text{ and } \bowtie\; \in \{<, \leq, =, \not=, \geq, > \}\}$. The intuition is that a proposition $(p \bowtie n)$ is true in a marking $M$ if the number of tokens in the place $p$ satisfies the proposition with respect to $n$. Since atomic propositions only depend on the discrete part of a marking, we adopt the same definition of $\mu$ for symbolic markings. For a state $s \in S$ and a formula $\phi$, we define the satisfaction relation $s \models \phi$ inductively as follows:
\begin{alignat*}{2}
s &\models (p \bowtie n)		&&\text{ iff } (p \bowtie n) \in \mu(s) \\
s &\models \neg \psi			&&\text{ iff } s \not\models \psi \\
s &\models \psi_1 \wedge \psi_2		&&\text{ iff } s \models \psi_1 \text{ and } s \models \psi_2 \\
s &\models \psi_1 \vee \psi_2 		&&\text{ iff } s \models \psi_1 \text{ or } s \models \psi_2 \\
s &\models \EF \psi			&&\text{ iff } s \goes{}^* s' \text{ and } s' \models \psi \\
s &\models \AG \psi			&&\text{ iff } s \not\models \EF \neg \psi \quad .\\
\end{alignat*}


As the $\AG$ and $\EF$ temporal operators are dual, it is enough to
design an algorithm for deciding the validity of $\EF \psi$.
Note that because the predicates do not allow us to query the ages
of tokens in the net, the presence of age invariants in the TAPN model
adds an expressive power (otherwise we could conjunct the age invariants
with the intervals on input arcs and add to the formulae the requirement that
no place contains any token exceeding the invariant bound).

We say that a place $p$ in a boolean predicate $\psi$
defined according to \cref{eqn:psi} is \emph{monotonicity-breaking}
if $\psi$ contains an atomic proposition of the form
$p < n$, $p \leq n$, $p = n$ or $p \not= n$. In other words, the 
inequality imposes some upper bound or an exact comparison to a
concrete number in the place $p$.

\begin{lemma}\label{lem:inc_empty_all_properties}
Let $M$ and $M'$ be symbolic markings and let $\psi$ be a boolean predicate 
defined by \cref{eqn:psi} and let
the set $\setOfIncPlaces$ of inclusion places do not contain any
monotonicity-breaking place. 
If $M \models \psi$ and $M \markingsubseteq{} M'$ then $M' \models \psi$.
\end{lemma}
\begin{proof}
By structural induction on $\psi$. The induction step
is trivial; we discuss here only the base case for a proposition of the formi
$\psi = p \bowtie n$.
Let $(\place{1},\Val{1})$ and $(\place{2},\Val{2})$ be symbolic markings 
such that $(\place{1},\Val{1}) \markingsubseteq{} (\place{2},\Val{2})$.
Let $(\place{1},\Val{1}) \models \psi$. 

If $p$ is a monotonicity-breaking place than $p \not\in \setOfIncPlaces$
and all tokens in the place $p$ belong to the set $\eq{(\place{1},\Val{1})}$.
By 
Condition \ref{def:inclusion-op-bijection-condition} of the inclusion
ordering
there exists a bijection $h$ such that 
for all $i \in \eq{(\place{1},\Val{1})}$ 
we have $\place{1}(i) = \place{2}(h(i))$ and hence
in the marking $(\place{2},\Val{2})$
the number of tokens in the place $p$ is equal to the number
of tokens in the place $p$ in the marking $(\place{1},\Val{1})$ and
we get $(\place{2},\Val{2})  \models \psi$.

If $p$ is not a monotonicity-breaking place, the constraint on $p$
has the form $p \geq n$ or $p > n$.
If the tokens in the place $p$ belong to $\eq{(\place{1},\Val{1})}$
we are done by the arguments as above.
If the tokens in the place $p$ belong to $\inc{(\place{1},\Val{1})}$
then by 
Condition \ref{def:inclusion-op-inc-count}
of the inclusion ordering 
the number of tokens placed in $p$ in the marking
$(\place{2},\Val{2})$ is at least the number of tokens in the marking 
$(\place{1},\Val{1})$ and because the proposition on $p$ states only
a lower-bound, we can again conclude that $(\place{2},\Val{2})  \models \psi$.
\qed
\end{proof}


In order to present an efficient reachability algorithm, we need a finite representation for the potentially infinite sets of valuations discussed in \cref{sec:discrete-inclusion}. We will thus use DBMs. However, we have to implement the operations used on the sets of valuations, such as delay, clock reset and intersection, on DBMs. Similarly, we need to define a DBM which represents a guard or invariant of the form $i \in \slbracket a, b \srbracket$ where $i$ is a clock and $\slbracket a, b \srbracket$ is a well-formed interval---here 
$\slbracket$ is either closed or open left parenthesis and
similarly for $\srbracket$.

\begin{proposition}
Let $D_1$ and $D_2$ be canonical DBMs over the clocks $\C$. Then the following operations and DBMs can be computed efficiently
\begin{enumerate}
\item (Delay) $D_1\up$ is a canonical DBM s.t. $[D_1\up] = [D_1]\up$.
\item (Reset) $D_1\reset{R}$ where $R \subseteq \C$ is a canonical DBM s.t. $[D_1\reset{R}] = [D_1]\reset{R}$.
\item (Intersection) $D_1 \cap D_2$ is a canonical DBM s.t. $[D_1 \cap D_2] = [D_1] \cap [D_2]$.
\item (Interval DBM) $D_{i \in \slbracket a,b \srbracket}$ is a canonical DBM s.t. $[D_{i \in \slbracket a,b \srbracket}] = \{ v \in \setOfVal \mid v(i) \in \slbracket a,b \srbracket \}$.
\item (Discrete Inclusion) Let $\place{1}$ and $\place{2}$ be placement functions. The expression $(\place{1}, [D_1]) \markingsubseteq{} (\place{2}, [D_2])$ can be computed efficiently.
\end{enumerate}
\end{proposition}
All these operations can be efficiently implemented for 
DBMs (see e.g. \cite{pettersson_phd, DBLP:conf/ac/BengtssonY03})
for details on operations 1--4; the fifth operator can be implemented using 
DBMs, as showed in the full version of the paper.


We can now perform a standard search through the state-space of symbolic
markings using the passed/waiting list approach. We start by
adding the initial marking to the waiting list. As long as the waiting
list is nonempty, a symbolic marking $M$ is removed from the waiting list,
added to the passed list, and all symbolic extrapolated successors of $M$
are explored. If a successor $M'$ of $M$ is below (w.r.t. the
ordering $\markingsubseteq{}$) some marking on the passed or waiting
list, we ignore it. Otherwise we add $M'$ to the waiting list and
remove from the waiting and passed lists all markings that are below $M'$.
We stop with a positive answer once we find a marking satisfying a property
we are searching for.
If the whole state-space is searched without finding such a marking, we
return a negative answer.
The search is performed only upto $k$ tokens in the net where this
number is supplied by the user (it is undecidable whether
there is some $k$ such that the net is $k$-bounded~\cite{memics_paper}).
If the net is $k$-bounded for the given $k$ (this can be automatically
verified) then this gives
a conclusive answer, otherwise the search can give a conclusive
answer only if it finds a marking satisfying the given property.

\begin{algorithm}[p]
\textbf{Name: } \textbf{succ($N, (\place{},D)$)}\;
\KwIn{A $k$-bounded TAPN $N$ and a symbolic marking $(\place{}, D)$.}
\KwOut{The set of successor markings for $(\place{}, D)$.}
\Begin{
	successors := $\emptyset$\;
	\ForAll{$t \in T$}{
		Let $\Delta := \{ i \in \{1, 2, \ldots,k \} \mid \place{}(i) \in \spre{t} \}$\;
		\ForAll{sets $\IN \subseteq \Delta$ where $\spre{t} = \{ \place{}(i) \mid i \in \IN\}$ and $\forall i \in \{ 1,2,\ldots,k \} \setminus \IN \, . \, (\place{}(i),t) \in \mathit{IA} \Rightarrow \mathit{Type}(\place{}(i),t) \neq \Inhib$}{
			Let $R := \{ i \in \IN \mid \mathit{Type}(\place{}(i), t) \neq \Transport{} \}$\;
			$\place{}'$ := move(\place{}, \IN, $t$)\;
			$D' := \left(D \cap \bigcap_{i \in \IN} D_{i\, \in \, c(\place{}(i),t)}\right)\reset{R} \cap \bigcap_{i\, \in \, \{1,\ldots,k\}} D_{i \in \iota(\place{}'(i))}$\;
			\If{$D'$ is consistent}{
			successors := successors $\cup \, \left\{ \left(\place{}', \extplace{\place{}'}{(D')\up  \cap \bigcap_{i \in \{1,2,\ldots,k\}} D_{i \in \iota(\place{}'(i))}}^c \right) \right\}$\;
			}
		}
	}
	\Return successors\;
}
\caption{Successor generation algorithm\label{alg:simple_successor_generation}}  
\end{algorithm}

\begin{algorithm}[p]
\textbf{Name:} \textbf{Reach($N, (\place{0}, \age{0}), \EF \psi$)}\;
\KwIn{A marked $k$-bounded TAPN $(N, (\place{0}, \age{0}))$, a formula 
$\EF \psi$ and a set $\setOfIncPlaces \subseteq P$ not containing any
monotonicity-breaking place in $\psi$.}
\KwOut{YES if $(\place{0}, \age{0}) \models \EF \psi$, NO otherwise.}
\Begin{ 
	$\mathit{PASSED} := \emptyset$\;
	Create DBM $D_0$ such that $[D_0] = \{ \age{0} \}$\;
	\lIf{$(\place{0}, \extplace{\place{0}}{D_0\up \cap \bigcap_{i \in \{1,2,\ldots,k\}} D_{i \in \iota(\place{0}(i))}}^c) \models \psi$}{ \Return YES\; }
	$\mathit{WAITING} := \{(\place{0}, \extplace{\place{0}}{D_0\up \cap \bigcap_{i \in \{1,2,\ldots,k\}} D_{i \in \iota(\place{0}(i))}}^c)\}$\;

	\While{$\mathit{WAITING} \neq \emptyset$}{
		Remove some $(\place{},D)$ from $\mathit{WAITING}$\;		
		$\mathit{PASSED} := \mathit{PASSED} \cup \{(\place{},D)\}$\;
	
			\ForAll{$(\place{}',D') \in \text{succ}(N, (\place{},D))$}
			{
				\If{\nllabel{alg_line:discrete_inclusion_subset}$\neg\exists (\place{}'',D'') \in \mathit{PASSED} \cup \mathit{WAITING} \, . \, (\place{}',[D']) \markingsubseteq{} (\place{}'',[D''])$}
	{				
						$\mathit{PASSED} := \mathit{PASSED} \setminus \{ (\place{}'',D'') \in \mathit{PASSED} \mid (\place{}'',[D'']) \markingsubseteq{} (\place{}',[D'])\}$\nllabel{alg_line:discrete_inclusion_superset_start}\;
						$\mathit{WAITING} := \mathit{WAITING} \setminus \{(\place{}'',D'') \in \mathit{WAITING} \mid (\place{}'',[D'']) \markingsubseteq{} (\place{}',[D'])\}$\; \nllabel{alg_line:discrete_inclusion_superset_end}
					\lIf{$(\place{}',D') \models \psi$}{ \Return YES\; }
					$\mathit{WAITING} := \mathit{WAITING} \cup \{ (\place{}',D') \}$\;
				}
			}
		}		
	\Return NO\;
}
\caption{Reachability algorithm\label{alg:reachability}}  
\end{algorithm}
The successor generation algorithm is presented in 
Algorithm~\ref{alg:simple_successor_generation} and the reachability 
algorithm is given in Algorithm~\ref{alg:reachability}.
Observe, as remarked above, that the algorithm will discard any generated successor marking if a larger marking is already present in the PASSED or WAITING list (line \ref{alg_line:discrete_inclusion_subset}). Similarly, if a generated successor marking is larger than some marking in the PASSED or WAITING list, then it will remove all such smaller markings from the PASSED and WAITING list (lines \ref{alg_line:discrete_inclusion_superset_start} to \ref{alg_line:discrete_inclusion_superset_end}). 


\begin{lemma}\label{lem:algorithm_termination}
Algorithm \ref{alg:reachability} terminates.
\end{lemma}
\begin{proof}
Let $N$ be a $k$-bounded TAPN. We must argue that the state-space of the symbolic semantics is finite. Since $N$ is a $k$-bounded TAPN, it follows that there are only finitely many placement functions. Further, from \cref{lem:finite_number_of_extrapolated_dbms} we know that there are only finitely many extrapolated DBMs for a given placement function. Thus, we may conclude that there are only finitely many symbolic markings in the symbolic semantics using the extrapolation operator. Observe that Algorithm \ref{alg:reachability} will add each symbolic marking to the $\mathit{WAITING}$ list at most once. Thus, it follows that the algorithm terminates.
\qed
\end{proof}

\begin{lemma}\label{lem:algorithm_yes_implies_initial_marking_satisfies_property}
If Algorithm \ref{alg:reachability} returns "YES", then $(\place{0},\age{0}) \models \EF \psi$.
\end{lemma}
\begin{proof}
Assume that Algorithm \ref{alg:reachability} returns "YES". We must show that $(\place{0},\age{0}) \goes{}^* (\place{},\age{})$ such that $(\place{}, \age{}) \models \psi$. We define $\abstrname{\mathit{ext}}(\place{}, [D]) \stackrel{\text{def}}{=} [\ext{D}]$ for any placement function $\place{}$ and canonical DBM $D$ (for any set of valuations $\Val{}$ that cannot be represented by a DBM we assume $\abstr{\mathit{ext}}{\place{},\Val{}} \stackrel{\text{def}}{=} \Val{}$).

Since the algorithm returned "YES", it must have found some symbolic marking $(\place{}, D)$ such that $(\place{},D) \models \psi$. Observe that Algorithm \ref{alg:reachability} (and Algorithm \ref{alg:simple_successor_generation}) will alternate between using the identity abstraction for discrete transition firings and $\abstrname{\mathit{ext}}$ for time delays. Thus, from the way the algorithm searches through the state-space, we may conclude that there must exist symbolic markings $(\place{1}, D_1), (\place{2},D_2), \ldots, (\place{},D)$ such that
\begin{align*}
(\place{0}, [D_0]) \symtrans{}_{\arbabstrsecond{\mathit{ext}}} (\place{1},[D_1]) \symtrans{}_{\arbabstrsecond{\mathit{ext}}} \cdots \symtrans{}_{\arbabstrsecond{\mathit{ext}}} (\place{}, [D])
\end{align*}
where $[D_0] = \{ v_0 \}$. By \cref{lem:extrapolation_below_abstraction_equiv} and \cref{thm:soundness_completeness}, we have that $\symtrans{}_{\abstrname{\mathit{ext}}}$ is sound. Thus, there exists a concrete marking $(\place{}, \age{})$ such that $(\place{0}, \age{0}) \goes{}^* (\place{},\age{})$ and $\age{} \in [D]$. Since atomic propositions depend only on the discrete part of a marking (placement function), it follows that $(\place{},\age{}) \models \psi$.\qed
\end{proof}

\begin{lemma}\label{lem:init_marking_satisfies_property_implies_algorithm_yes}
If $(\place{0},\age{0}) \models \EF \psi$ then 
Algorithm \ref{alg:reachability} returns "YES".
\end{lemma}

\begin{proof}
Assume that $(\place{0},\age{0}) \models \EF \psi$. This means that 
$(\place{0},\age{0}) \goes{}^* (\place{},\age{})$ and $(\place{}, \age{}) \models \psi$. We must show that Algorithm \ref{alg:reachability} 
returns "YES".
We define $\abstrname{\mathit{ext}}(\place{}, [D]) \stackrel{\text{def}}{=} [\ext{D}]$ as before.
By \cref{lem:extrapolation_below_abstraction_equiv} and \cref{thm:soundness_completeness}, we get that $\symtrans{}_{\abstrname{\mathit{ext}}}$ is complete. Thus, there exists a symbolic marking $(\place{},[D]) \models \psi$ such that $(\place{0}, [D_0]) \symtrans{\epsilon}_{\abstrname{\mathit{ext}}} \circ \, \, (\symtrans{T}_{\abstrname{\mathit{id}}} \circ \symtrans{\epsilon}_{\abstrname{\mathit{ext}}})^* \,\, (\place{}, [D])$ where $[D_0] = \{ v_0 \}$ and $\age{} \in [D]$.

We will now argue that Algorithm~\ref{alg:reachability} will find a symbolic marking  $(\place{}', D')$ such that $(\place{},[D]) \markingsubseteq{} (\place{}', [D'])$. It is easy to see that Algorithm~\ref{alg:reachability} together with Algorithm~\ref{alg:simple_successor_generation} implements a symbolic exploration of the form $\symtrans{\epsilon}_{\abstrname{\mathit{ext}}} \circ \, \, (\symtrans{T}_{\abstrname{\mathit{id}}} \circ \symtrans{\epsilon}_{\abstrname{\mathit{ext}}})^*$. 
However, notice that the algorithm discards some of the discovered symbolic markings (lines \ref{alg_line:discrete_inclusion_subset} to \ref{alg_line:discrete_inclusion_superset_end} in Algorithm~\ref{alg:reachability}). If the algorithm finds a symbolic marking $(\place{}',D')$ for which $(\place{}',[D']) \markingsubseteq{} (\place{}'',[D''])$ for some $(\place{}'',D'')$ in the PASSED or WAITING list, it will discard $(\place{}',D')$ (line \ref{alg_line:discrete_inclusion_subset}). Similarly, if $(\place{}'',[D'']) \markingsubseteq{} (\place{}',[D'])$ for some $(\place{}'',D'')$ in the PASSED or WAITING list, it will remove all markings $(\place{}'',[D'']) \markingsubseteq{} (\place{}',[D'])$ from both the PASSED and WAITING list (lines \ref{alg_line:discrete_inclusion_superset_start} to \ref{alg_line:discrete_inclusion_superset_end}). However, by \cref{thm:monotonicity_theorem} it is safe to skip these symbolic markings since the future behaviour of the smaller symbolic markings is included in the larger symbolic marking. 
Thus, it follows that Algorithm~\ref{alg:reachability} will find a symbolic marking $(\place{}',D')$ such that $(\place{},[D]) \markingsubseteq{} (\place{}', [D'])$. 
By \cref{lem:inc_empty_all_properties}, we have that if the smaller marking satisfies $\psi$ then the larger marking $(\place{}',[D'])$ also satisfies $\psi$.  
Thus, Algorithm~\ref{alg:reachability} returns "YES".
\qed
\end{proof}

The correctness of the reachability algorithm is hence established.

%% file: experiments-short.tex
\begin{table}[t]
\begin{center}
\begin{tabular}{|c|r|r|r|r|}
\hline
Delay
   & TAPAAL & TAPAAL incl. & Broadcast & Deg.2 Broadcast\\
\hline
28 & 10.8 & 10.4 & 11.6 & 11.6\\
24 & 12.1 & 12.0 & 102.3 & 48.8\\
20 & 17.0 & 16.4 & 456.2 & 88.0\\
16 & 92.6 & 90.7 & 207.4 & 137.7\\
\hline
\end{tabular}
\end{center}
\vspace{-6mm}
\caption{PMS case study scaled by the sampling delay (time in seconds)}
\label{fig:PMS}
\end{table}

\begin{table}[t]
\begin{center}
\begin{tabular}{|c|r|r|r|r|}
\hline
Messages
   & TAPAAL & TAPAAL incl. & Broadcast & Deg.2 Broadcast\\
\hline
2 &  2.5 & 0.6  & 2.9  & 2.3  \\
3 &  11.6 & 2.1 & 12.0  & 7.8  \\
4 &  46.3 & 8.0  & 46.2 & 24.9 \\
5 &  164.1 & 29.1  & 165.0  & 73.5  \\
6 &  $>$400.0 & 109.6  & $>$400.0  & 197.7  \\
7 &  $>$400.0 & 330.4  & $>$400.0  & $>$400.0  \\
\hline
\end{tabular}
\end{center}
\vspace{-6mm}
\caption{BAwPC scaled by number of retransmission messages 
(time in seconds)}
\label{fig:BAwPC}
\end{table}

We implemented the reachability algorithm in {\tt C++} and 
fully integrated it into the tool TAPAAL~\cite{DJJJMS:TACAS:12} 
(www.tapaal.net), an 
open-source and platform-independent editor, simulator
and verifier of extended timed-arc Petri nets. In order to document
the performance of our proposed algorithm, we present two larger case studies
of Patient Monitoring System (PMS) and a communication protocol 
from the WS-Business Activity standard~\cite{wsba}. Both models
can be downloaded from the tool's homepage.

The patient monitoring system (PMS) is a case study
taken from~\cite{TSPN-PMS:2012}. The system monitors the pulse rate
and the level of oxygen saturation via sensors applied on
the skin of a patient. It consists of three components:
sampling subsystem, signal analyzer and alarm. The purpose
of the PMS model is to verify that abnormal situations dangerous
for the patient's health are detected within given deadlines.
We have verified the model for deadline violation both in the
sampling component and the signal analyzer. The sampling delay 
has been varied from 28 down to 16 seconds.
This increased the complexity of the verification, as the queries
were still satisfied and the whole state-space had to be searched.

In the second case study we verify the correctness of one of the
web services coordination protocols called
Business Activity with Participant Completion (BAwPC)~\cite{wsba}.
Our model is based on the work presented in~\cite{WSBA:verification}
where an enhanced protocol that avoids reaching any invalid states
is given. We modelled the protocol in TAPAAL and considered
asynchronous communication where messages can be lost; the model
is scaled by the number of extra messages that can be used for
retransmissions. The protocol is correct and hence the whole state-space is 
searched. 

We compare the performance of our implementation with the UPPAAL engine where
the timed automata models were obtained by automatic translations
(called broadcast and degree 2 broadcast; 
see~\cite{JJMS:EPEW:10,JJMS:SOFSEM:11} for the details) 
from the TAPN models. 
We remark that the verification times of the translated TAPN models are 
in general comparable with native UPPAAL models and in some examples
the translated models verify even faster than the native
ones~\cite{JJMS:SOFSEM:11}. A direct comparison with other Petri net
tools extended with time like Romeo~\cite{romeo} and 
TINA~\cite{tina} is not possible due to the
radically different semantics of the Petri net models used in these tools. 

All experiments were run on Macbook Pro with 2.7 GHz 
Intel Core i7 with 8 GB RAM using BFS search strategy and the results 
are presented
in Tables~\ref{fig:PMS} and~\ref{fig:BAwPC}.
The column
TAPAAL refers to our algorithm where the set of inclusion places
has been set to empty and TAPAAL incl. is our algorithm with 
the largest possible inclusion set. The user has the possibility
to choose between these algorithms (or even manually select the
concrete inclusion places) because for example in the case of 
1-safe Petri nets where the inclusion is only rarely applied,
the algorithm with the maximum inclusion can be slower due to the
implementation overhead connected with inclusion checking of markings
on the passed and waiting list. Indeed, in situations 
like in Table~\ref{fig:PMS} the full inclusion checking is not
that beneficial opposite to nets like in Table~\ref{fig:BAwPC} where
we have many tokens (messages) in the same place.

%% file: conclusion.tex
We presented a reachability algorithm for extended timed-arc Petri nets
and implemented it within the tool TAPAAL.
The algorithm includes efficient extrapolation and 
symmetry reduction techniques that show a very encouraging
performance even on larger
case-studies. We would like
to emphasize the fact that all features that are implemented
in the tool are formally defined and proved correct.
We believe that our tool, available at www.tapaal.net, 
is one of a rather few reasonably-sized model 
checkers with a complete correctness proof taking into account all 
implemented optimizations and reduction techniques.
In the future work we shall look at extending the technique
to liveness properties and at further performance improvements 
by using for example the 
LU-extrapolation~\cite{DBLP:journals/sttt/BehrmannBLP06}.

\paragraph{Acknowledgements.} We would like to thank the anonymous
reviewers for their comments.